\documentclass{ifacconf}

\usepackage{graphicx}      
\usepackage{natbib}        
\usepackage{algorithmic}
\usepackage{graphicx}
\usepackage{textcomp}
\usepackage{xcolor}
\usepackage{amssymb}
\usepackage[centertags]{amsmath}
\usepackage{relsize}
\usepackage{textfit}
\usepackage{dcolumn}
\usepackage{url}
\usepackage{tikz,pgfplots}
\usetikzlibrary{arrows,decorations.pathreplacing,intersections,calc}
\usepackage{slashbox}
\usepackage{comment}
\usepackage{siunitx}
 \usepackage{url}

\usepackage[width=0.9\textwidth,labelfont=bf]{caption}
    \usepackage{graphicx}      
\usepackage{balance}        

\usepackage{algorithmic}
\usepackage{graphicx}
\usepackage{textcomp}
\usepackage{xcolor}
\usepackage{amssymb}
\usepackage[centertags]{amsmath}
\usepackage{relsize}
\usepackage{textfit}
\usepackage{dcolumn}
\usepackage{url}
\usepackage{tikz,pgfplots}
\usetikzlibrary{arrows,decorations.pathreplacing,intersections,calc}
\usepackage{slashbox}
\usepackage{comment} 
\usepackage{siunitx}
 \usepackage{url}
\usepackage{multirow}

\usepackage{tabularx}

\usepackage{exscale}

\usepackage{subfigure}
\usepackage{arydshln}

\usepackage{upgreek}

\usepackage{bm}

\usepackage{times,mathptmx}
\usepackage[T1]{fontenc}

\usepackage{cancel}

\newcommand{\norm}[2]{\left \lVert #1 \right \rVert_{#2}}

\usepackage{exscale}

\usepackage{subfigure}
\usepackage{arydshln}

\DeclareMathAlphabet{\mathcal}{OMS}{cmsy}{m}{n}

\newtheorem{definition}[thm]{Definition}
\newtheorem{infdef}[thm]{Informal Definition}

\usepackage{upgreek}

\usepackage{bm}

\usepackage{times,mathptmx}
\usepackage[T1]{fontenc}

\usepackage{cancel}

\begin{document}
\begin{frontmatter}

\title{Online Learning-Based Control with Guaranteed Error Bounds for a Class of Nonlinear Systems}

\author[First]{Ricus Husmann} 
\author[First]{Sven Weishaupt} 
\author[Second]{Malin Lotta Husmann} 
\author[First]{Harald Aschemann}

\address[First]{Chair of Mechatronics, University of Rostock, 
    Germany\\ \{Ricus.Husmann, Sven.Weishaupt, Harald.Aschemann\}@uni-rostock.de}
\address[Second]{Faculty of Physics, Dresden University of Technology, 
    Germany\\ malin\_lotta.husmann@mailbox.tu-dresden.de}

\begin{abstract}                
In this paper, we present a learning-based control for a class of nonlinear systems that guarantees exponential stability as well as bounded output errors. The control is based on the Gaussian Process Submodel Online Learning (GPSOL) algorithm and the Disturbance Error Rate Limiting (DERL) algorithm, both of which were developed in previous work. The GPSOL algorithm provides a method to learn Gaussian Process (GP) models for subsystems online, whereas the DERL filter allows to limit the rate of the prediction error of these GP models. The focus of this paper is the utilization of the GP model within an adaptive controller and the derivation of corresponding stability conditions and system peak-to-peak gains by means of linear matrix inequalities (LMIs). These peak-to-peak gains are then used to prescribe a desired prediction error rate for the DERL filter to achieve user-defined output error bounds. The gains and the related bounds were successfully verified using a simulation model. Furthermore, results form a successful experimental validation of the bounds and the overall control structure on a pneumatic test rig are presented.  While the control scheme and error bounds proposed in this paper are limited to first-order single-input-single-output systems, an extension to certain classes of higher-order and  multiple-input-multiple-output  systems is expected to be forthcoming.
\end{abstract}


\end{frontmatter}

\section{Introduction}
Typically, in the modeling of technical systems certain submodels (models for  subsystems or model parts) are present that are easy to describe by means of first-principles models, whereas other submodels are hard to derive. In the latter case, the reasons for these \textit{elusive} characteristics  might be an inherent complexity or a lack of knowledge. Examples of systems that can be divided in these two classes of submodels involve mechanical systems with friction or thermofluidic systems. While the equations of motion governing a mechanical system are often easily to derive, accurate friction models are hard to obtain by first principles, see \cite{Scholl:2024}. Similarly, in thermofluidic systems, heat transfer coefficients can often be derived only roughly from first principles, whereas energy and mass balances equations can be easily established. One outcome would be to use model-free control methods for such systems, like PID controllers or Active Disturbance Rejection Control (ADRC), see \cite{Chen:2016}. Alternatively, data-driven approaches like Reinforcement Learning (RL), see \cite{Chen:2025}, or various data-driven MPC variants may be applicable, see e.g. \cite{Berberich:2021}. While all these model-free methods can be very potent, there seems to be a reasonable trend towards combining existing model knowledge with learning-based methods  to improve both learning speed and sample efficiency -- like in Physically Informed Neural Networks (PINNs), cf. \cite{Cuomo:2022}, or Deep Lagrangian Networks (DLNs), cf. \cite{Lutter:2019}. This fundamental idea of combining prior knowledge with learning-based methods also motivates and applies to our algorithm. Instead of trying to encode knowledge into machine learning methods, we split the system into known subsystems as well as elusive subsystems and represent them accordingly by means of first-principles models and online learning-based methods. Furthermore, we aim to utilize both kind of models within a model-based control structure. For the sake of simplicity, we generally focus in this paper on systems with only a single elusive submodel governed by a single unknown function. Most of the developed methods are, however, expected to be extendable to multiple functions and submodels. In a Gaussian Process (GP) setting, the function to be learned  is often referred to as the {\it hidden function}.

One approach to handle the previously described kind of systems is known as regression and involves the definition of parameterized ansatz functions for the elusive submodels, e.g. using polynomials and a subsequent parameter identification. This can also be done in an online fashion by means of nonlinear Kalman Filter variants, see \cite{Julier:1997}, or moving horizon estimation, see \cite{Haseltine:2005}. If the problem-related function type of the elusive submodel is known, these approaches can be very powerful. If this knowledge is not available, which is quite often the case in practice, these classical methods still leave room for improvement. The usage of Neural Networks (NNs) in the form of Multi-Layer Perceptrons (MLPs), Radial Basis Functions (RBFs) or Deep NNs (DNNs) as a generalization of parametric ansatz functions, could provide a solution here. They are, however, generally subject to the problem of overfitting, which is especially severe in the online learning variants. Here, Gaussian Processes (GPs) represent a powerful alternative tool for data-driven modeling, because they offer flexibility, robustness against overfitting and, especially, a built-in uncertainty measure, see \cite{Schuerch:2020}. The disadvantage of using GPs for online-learning is given by the fact that it does not scale well in terms of data as a non-parametric approach. There exist numerous methods in the literature to tackle this problem. One promising approach is the usage of GPs as parametric functions within Recursive Gaussian Process regression proposed by \cite{Huber:2013}. This method preserves many of the advantages while removing the GP scaling issues w.r.t.~data. For our particular application, the basic RGP algorithm has a major drawback, because it needs a direct measurement of the GP output. Therefore, we derived the  Gaussian Process Submodel Online Learning (GPSOL) algorithm in \cite{RGP_dKF:2025,RGP_Journal:2025}. Furthermore, an also real-time capable algorithm for the inclusion of monotonicity assumptions was presented in \cite{RGP_mon:2025}. Both the GPSOL algorithm and the monotonicity inclusion were successfully experimentally validated for a Vapor Compression Cycle and several pneumatic systems, see \cite{RGP_dKF:2025,RGP_Journal:2025}.

Learning-based control adaptive control have a long history. Classic examples include adaptive control, cf.~\cite{Astroem:2008}, and, especially, Model Reference Adaptive Control (MRAC), see \cite{Nguyen:2018}. Nonlinear systems with uncertainties or disturbances may be controlled using adaptive backstepping, see the control application presented in \cite{Wache:2021}. Most often, only single parameters or lumped variables are learned or estimated, e.g. using disturbance observers, see e.g. \cite{Aschemann:2009}. Nevertheless, there are examples, e.g. \cite{Calise:1998}, where nonlinear algebraic relations are learned and subsequently used in control. In \cite{Sena:2021}, an Exteded Kalman Filter (EKF) is used to learn a NN for model predictive control (MPC). There also exist several implementations for GP models in control: The main focus is here the combination with stochastic MPC, like in \cite{Bradford:2020} and, recently, in \cite{Wietzke:2025}. The reason is that the covariance prediction, which is an intrinsic part of GPs, can be effectively used within the stochastic MPC framework. This uncertainty prediction is also used in the Disturbance Error Rate Limiting (DERL) algorithm to limit the error rate of the GP prediction.

The paper is structured as follows: First, we give a brief overview of the previously presented GPSOL algorithm and DERL filter in Sec.~\ref{sec:Recap}. Then, we derive the adaptive, learning-based controller in Sec.~\ref{sec:Control} and the resulting error dynamics of the controlled system. In Sec.~\ref{sec:error_bounds}, we propose the stability conditions and output error bounds for a general class of systems with a bounded nonlinearity, and in the following Sec.~\ref{sec:usage_of_bounds}, we show how these output error bounds apply to the the error dynamics of the controlled system under consideration. In Sec.~\ref{sec:ProposedUtilization}, we present one potential method of usage of the derived bounds. The error bounds are  numerically verified in Sec. \ref{sec:numval} and experimentally validated in Sec. \ref{sec:exval}. Finally, we finish the paper with conclusions and an outlook. The main contributions of this paper are given by the following derivations and results:
\begin{itemize}
\item a learning-based control for a class of nonlinear system using Gaussian Processes 
\item stability conditions and peak-to-peak gains for the controlled system (based on LMIs) that allow for the quantification of control output error bounds
\item numerical verification and experimental validation of the derived bounds and the overall control structure.
\end{itemize}

\section{Previous Results}
\label{sec:Recap}
In the following, we will briefly present some previous findings and definitions that are essential for the presented control.
\subsection{Gaussian Process Submodel Online Learning (GPSOL)}
\label{sec:GPSOL}
In the basic GPSOL algorithm published in \cite{RGP_dKF:2025}, we consider nonlinear discrete-time systems 
\begin{equation}
\bm{x}_{k+1}=\bm{f}(\bm{x}_k,\bm{u}_k,z_k)+\bm{\nu}_k \,,~ \bm{y}_k=\bm{h}(\bm{x}_k) +\bm{\epsilon}_k  \,, 	\label{eq:NLinDiffEqGPSOL}
\end{equation}
with states $\bm{x}_k\in\mathbb{R}^{n_x}$, Gaussian process noise $\bm{\nu}_k \sim \mathcal{N}(\bm{0},\bm{Q}_x) $  and control inputs $\bm{u}_k\in\mathbb{R}^{n_u}$, noisy measurement outputs $\bm{y}_k\in\mathbb{R}^{n_y}$ with the Gaussian measurement noise $\bm{\epsilon}_k \sim \mathcal{N}(\bm{0},\bm{R})$ and a single, time-invariant hidden function $z_k=z_f(\bm{\zeta}_k)$ that depends only on $n_z$ deterministic inputs $\bm{\zeta}_k=[\zeta_{k,1},\zeta_{k,2},\dots, \zeta_{k,n_z}]^T$. These deterministic inputs into the hidden function could be, for example, measurable external disturbances. As convention, the index $k$ denotes the timestep which is linked to the time $t=k T_s$ by means of the sampling time $T_s$.

The GPSOL algorithm provides a method to learn a GP model for the hidden function  $z_f(\bm{\zeta}_k)$ online by recursively updating the mean values and covariance for a set of user-defined basis vectors. Naturally, the defined resolution of the basis vector grid marks a trade-off between the achievable approximation quality of the hidden function and the corresponding computational effort. Several experimental validations showed, however, that very satisfying results can be achieved for mechatronic systems with sampling rates in the range of milliseconds. We assume that the hidden function as well as its inputs are bounded and also their rates.

\begin{infdef}
For systems of type \eqref{eq:NLinDiffEqGPSOL}, the GPSOL algorithm learns a GP model $\tilde{z}_{f,k}$ of a hidden function ${z}_f$. The learned GP model provides a Gaussian prediction  $\tilde{z}_k \sim \mathcal{N}({\mu}_{k}^I, {c}_{k}^I)$. The prediction noise of the GP model is bounded to the interval  $\sigma_r^2\leq {c}_{k}^I\leq \sigma_K^2$, where $\sigma_K$ is the vertical hyperparameter, which also marks the initial uncertainty of the GP model.  $\sigma_r<<\sigma_K$ represents the remaining uncertainty defined by the user, which is an approximate measure of the highest possible approximation quality of the finite dimensional GP w.r.t. the hidden function. Assuming that the hidden function is identifiable with the available measurements, the GP model provides an approximation of the hidden function: \newline $\tilde{z}_k \rightarrow {z}_k + \nu_{r,k}$ with $\nu_{r,k} \sim \mathcal{N}(0, \sigma_r^2)$ for $k \rightarrow \infty$.
\end{infdef}

\subsection{Previously Derived Bounds}
\label{sec:prevBoudns}
In \cite{RGP_Journal:2025}, we consider a class of input-output-controllable single-input single-output (SISO) systems that can be transformed into the following normal form 
\begin{equation}
\resizebox{\columnwidth}{!}{$
\underbrace{\frac{d}{dt} \left(\left[\begin{matrix} y \\ \dot{y} \\ \ddot{y} \\ \vdots\\  \overset{(n)}{y} \end{matrix}\right] \right)}_{\dot{\bm{y}}(t)}=  \underbrace{\left[\begin{matrix} 0&1&0 & 0&\dots\\ 0&0&1 &0& \dots\\ 0&0&0 &1& \dots  \\ &\vdots&& \\ 0& 0&0&0& \dots \end{matrix}\right]}_{\bm{A}} \underbrace{\left[\begin{matrix} y \\ \dot{y} \\ \ddot{y} \\ \vdots\\  \overset{(n)}{y} \end{matrix}\right]}_{\bm{y}(t)} +\left[\begin{matrix}0\\0\\0\\ \vdots \\a(\bm{y}(t),t) \end{matrix}\right] + \left[\begin{matrix}0\\0\\0\\ \vdots \\b(\bm{y}(t),t) \end{matrix}\right] u+ \left[\begin{matrix}0\\0\\0\\ \vdots\\ e \end{matrix}\right] z(t)$
} \label{eq:OldSystem}
\end{equation}
with state vector $\bm{y}(t)$, scalar input $u(t)$, scalar output $y(t)$, disturbance $z(t)=z_f(\bm{\zeta}(t))$ and the constant disturbance gain $e$. This system can be controlled by the following input-output linearization-based control law that involves an integral part
\begin{equation}
u(t)=\frac{1}{b(\bm{y}(t),t)} \left(\upsilon(t) - a(\bm{y}(t),t) - e\hat{z}(t)\right) \,, \label{eq:OldControl}
\end{equation}
with $\upsilon(t)= \overset{(n+1)}{y_r}(t) +\bm{k}^T \left(\bm{y}_r(t)-\bm{y}(t) \right)+k_I \int \left(y_r(t)-y(t) \right)dt$ and the DERL- filtered estimate for the disturbance $\hat{z}(t)$. The constant vector $\bm{k}$ and the constant $k_I$ correspond to the feedback gains of a linear proportional-integral state feedback, and the vector $\bm{y}(t)$ refers to the measured values for the control output and its time derivatives up to the system order $n$.

As shown in \cite{RGP_Journal:2025} for vanishing initial conditions and purely real-valued eigenvalues of the controlled system, an output bound can be provided 
\begin{equation}
\norm{\dot{z}(t)-\dot{\hat{z}}(t) }{\infty} \frac{|e|}{k_I} \geq \norm{y_r(t)-y(t)}{\infty}  \,, \label{eq:inequality3}
\end{equation}
which only depends on the rate limit on the GP prediction error rate $\norm{\dot{z}(t)-\dot{\hat{z}}(t)}{\infty}$, the integrator gain $k_I$ and disturbance input gain $e$.

Furthermore, \eqref{eq:inequality3} also gives worst-case guarantees according to
\begin{equation}
(\norm{\dot{z}(t)}{\infty}+\norm{\dot{\hat{z}} (t)}{\infty})\frac{|e|}{k_I} \geq \norm{y_r(t)-y(t)}{\infty}  \,. \label{eq:inequality4}
\end{equation}
Thus, if the actual rate of the hidden function output $\dot{z}(t)$ is known to be bounded, a bound on the output error is given by artificially rate-limiting the absolute value of the GP prediction $|\dot{\hat{z}}(t)|\leq z_{lim}$.

\subsection{Disturbance Error Rate Limiting (DERL)}
\label{sec:DERL}
The results presented in Subsec. \ref{sec:prevBoudns} motivated the introduction of the so-called DERL  algorithm in \cite{RGP_Journal:2025}, which utilizes the uncertainty information provided by the GP model to intelligently filter the GP prediction. Please note that both the GPSOL and DERL filter are derived with discrete-time formulations. The GP prediction rates are approximated by $\dot{\tilde{z}}(t) \approx (\tilde{z}_{k+1}-\tilde{z}_k)/T_s$ and $\dot{\hat{z}}(t) \approx (\hat{z}_{k+1}-\hat{z}_k)/T_s$.

\begin{definition}
\label{def:DERL}
\textit{The scalar  $\hat{z}(t)$  is bounded and a filtered version of the GPSOL output $\tilde{z}(t)$ which satisfies for all time steps $k=t/T_s$ one of the cases:
\begin{enumerate}
\item the inequality $|\dot{z}_e (t)| \leq  z_{lim}$ with a probability of $p_{lim}$ 
\item the inequality $|\dot{\hat{z}}(t)|\leq  z_{lim}$ deterministically.
\end{enumerate}
Here, the bound is a constant, positive parameter $z_{lim}\in \mathbb{R}^+$ and $\dot{z}_e(t)=\dot{\hat{z}}(t)-\dot{z}(t)$ denotes the prediction error rate. $z(t)$ will be formally defined in Def. \ref{def:sys}. The probability $p_{lim}$ is user-defined.
}
\end{definition}
\begin{cor}
\label{cor:DERL}
Def. \ref{def:DERL} guarantees $\norm{\dot{z}_e(t)}{\infty}\leq z_{lim}+\norm{\dot{z}(t)}{\infty} $ with probability $p_{t}$. 
\end{cor}
 \begin{pf}
 For case (2), the respective inequality can be transformed to $|\dot{\hat{z}}|+|\dot{z}|\leq  z_{lim}+|\dot{z}|$. Since it holds that $|\dot{z}_e|=|\dot{\hat{z}}-\dot{z}|\leq |\dot{z}|+|\dot{\hat{z}}|$, case (2) guarantees Cor. \ref{cor:DERL}. Since  $z_{lim}\leq z_{lim}+|\dot{z}|$, case (1) also guarantees Cor. \ref{cor:DERL}. For the determination of $p_t$, please refer to Cor. 17 in \cite{RGP_Journal:2025}. 
  \qed 
 \end{pf}
\begin{rem}
As elaborated in Remark 18 in \cite{RGP_Journal:2025}, $p_{t}\approx 1$ might be assumed for most practical applications if $p_{lim}$ is set sufficiently high.
\end{rem}
For a system defined in Sec. \ref{sec:prevBoudns} the Cor. \ref{cor:DERL} guarantees the bound		\begin{equation}
(\norm{\dot{z}(t)}{\infty}+z_{lim})\frac{|e|}{k_I} \geq \norm{y_r(t)-y(t)}{\infty}  \,. \label{eq:inequality5}
\end{equation}
Naturally, this output error bound on $\norm{y_r-y}{\infty} $ is only guaranteed in the particular case of system \eqref{eq:OldSystem} together with controller \eqref{eq:OldControl}. The main goal of this paper is, hence, the relaxation regarding the assumption of a constant disturbance input gain $e$, which allows for the application to a wider class of controlled systems.

\section{Control Design}
\label{sec:Control}
In this section, we develop the control structure depicted in Fig.~\ref{pic:RGP_Control} and derive the resulting error dynamics of the controlled system. As defined in Subsec.~\ref{sec:adapcontr} and Subsec.~\ref{sec:ex_model_learning}, the output $z(t)$ of the hidden function is estimated by both the adaptive controller, which provides $\overline{z}(t)$, and by the GP model learned by GPSOL, which provides the estimate $\tilde{z}(t)$. This GP prediction is filtered afterwards by means of the DERL filter to provide the filtered prediction $\hat{z}(t)$. The DERL filter is parameterized using the peak-to-peak gain $\gamma$ according to Sec.~ \ref{sec:usage_of_bounds}.
\begin{figure}
	 \begin{center}
		 \includegraphics[width=1\linewidth]{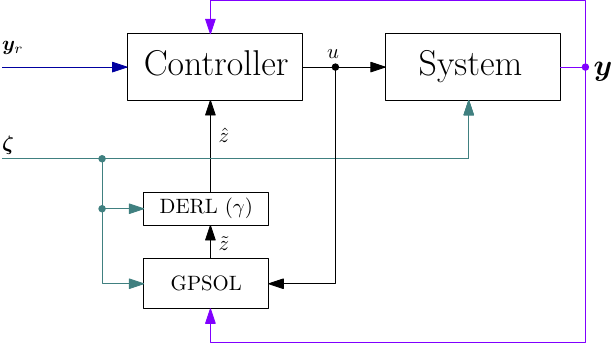}
		  \caption{Schematic structure of the control scheme including the GPSOL algorithm and DERL filter.} 
		  \label{pic:RGP_Control}
	 \end{center}
\end{figure}

\subsection{System}
\label{sec:System}
\begin{definition}
\label{def:sys}
\textit{ For the scalar control input $u(t)$,  the scalar measurable control output $y(t)$, a nonlinear, input-affine, first-order, single-input single-output (SISO) system is given by the following ODE
\begin{equation} 
\dot{y}(t)=a(y,t)+b(y,t)u(t)+e(y,t)z(t) \,, \label{eq:SystemBase}
\end{equation}
 with the Lipschitz-continuous gains $a(y(t),t)$, $b(y(t),t)$ and $e(y(t),t)$, where $b(y(t),t)\neq 0$ holds.  We define the disturbance as some Lipschitz-continuous function
  \begin{equation}
  z_f: \bm{Z}_{in} \rightarrow \mathbb{R},\bm{\zeta} \rightarrow z_f(\bm{\zeta}) , z_f \in C^1(\bm{Z}_{in})\,.
 \end{equation}
 Here, the inputs $\bm{\zeta}(t)$ are perfectly measurable and their set $\bm{Z}_{in} \subset \mathbb{R}^{n_z}$  is bounded. Furthermore, their rates $\dot{\bm{\zeta}}(t)$ are bounded. To abbreviate the notation, we define $z(t)= z_f(\bm{\zeta})$.}
\end{definition}
\begin{cor}
The disturbances $\{z(t)\}$ and their rates $\{\dot{z}(t)\}$ are bounded.
\end{cor}
 \begin{pf}
 This follows directly from Lipschitz continuity and continuous differentiability of  $z_f$,  and the boundedness of its inputs $\bm{\zeta}$ and their rates $\dot{\bm{\zeta}}$ under application of the extreme value theorem.  \qed
 \end{pf}

\subsection{Adaptive Control Design}
\label{sec:adapcontr}
\begin{prop}
For positive constant gains $K_P>0$ and $K_I>0$ as well as constant disturbances according to $\dot{z}=0$, the following control law
\begin{equation}
u(t)= \frac{1}{b(y,t)}\left(\dot{y}_r(t)+K_P e_y(t) -a(y,t)-e(y,t) \overline{z}(t) \right), \label{eq:ControlLaw}
\end{equation}
with the control error $e_y(t)=y_r(t)-y(t)$  w.r.t.~the reference $y_r(t)$, and the parameter update law
\begin{equation}
\dot{\overline{z}}(t) =- K_I e(y,t) e_y(t) \label{eq:ParamUpdateLaw}
\end{equation}
for the parameter estimate $\overline{z}(t)$ exponentially stabilize the system defined in Def. \ref{def:sys}.
\end{prop} 
 \begin{pf}
 The stability proof is given by Cor. \ref{cor:appl4}. \qed
 \end{pf}
\begin{rem}
For $\dot{z}=0$, this controller is equivalent to an adaptive backstepping controller for the first-order system, see \cite{Krstic:1995,khalil:2002} for an overview.
\end{rem}

\subsection{Extension and Model learning}
\label{sec:ex_model_learning}
In the following, we remove the assumption $\dot{z}=0$ and return to the original assumptions regarding $z(t)$ as given in Subsec.~\ref{sec:System}. Furthermore, we now utilize the GPSOL algorithm described in Subsec. \ref{sec:GPSOL} together with the DERL filter defined in Def. \ref{def:DERL} to provide a learned and filtered estimate $\hat{z}(t)$ for the disturbance.  We use this estimate to extend the control law \eqref{eq:ControlLaw} to
\begin{equation}
u= \frac{1}{b(y,t)}\left(\dot{y}_r+K_P e_y -a(y,t)-e(y,t) \left(\overline{z}(t)+\hat{z}(t)\right) \right) \label{eq:ExControlLaw}
\end{equation} 
and introduce the error of this estimate $z_e(t)=\hat{z}(t)-z(t)$.

\begin{prop} \label{prop1} The system \eqref{eq:SystemBase}  controlled by \eqref{eq:ExControlLaw} and \eqref{eq:ParamUpdateLaw} can be described by means of the following quasi-linear error dynamics
\begin{equation}
\dot{\bm{x}}(t)=\underbrace{\left[ \begin{matrix}
-K_P& \overline{e}(e_y,y_r,t) \\
-K_I \overline{e}(e_y,y_r,t)&0 
\end{matrix}\right]}_{\bm{A}(\bm{x}(t),t)}\bm{x}(t) + \underbrace{\left[ \begin{matrix} 0 \\
1 
\end{matrix}\right]}_{\bm{b}} v(t)  \label{eq:QlinSysDyn}
\end{equation}
with the input $v(t)=\dot{z}_e(t)$ and the state vector  $\bm{x}(t)=[e_y(t),e_z(t)]^T$, where $e_y(t)=y_r(t)-y(t)$ and $e_z(t)=z_e(t)+\overline{z}(t)$ and the new disturbance input gain $\overline{e}(e_y(t),y_r(t),t)=e(y(t),t)$ with $y(t)=y_r(t)-e_y(t)$.
\end{prop}

 \begin{pf}
 If we evaluate the system dynamics \eqref{eq:SystemBase} with the control law \eqref{eq:ExControlLaw} and the parameter update law \eqref{eq:ParamUpdateLaw}, we obtain the following ODEs describing the controlled system
 \begin{align}
 \dot{y}&=\dot{y}_r+K_P e_y+e(y,t) \underbrace{\left(z-\hat{z}-\overline{z}\right)}_{-e_{z}},  \\ \label{eq:SyDyn2}
 \dot{\overline{z}} &=- K_I e(y,t) {e}_y \,.
 \end{align}
 This system can be reformulated and yields the error dynamics by utilizing the new estimation error $e_{z}=-z+\hat{z}+\overline{z}=z_e+\overline{z}$. With the corresponding time derivative  $\dot{e}_{z}=\dot{z}_e+\dot{\overline{z}}$, this results in the following ODEs for the error system
 \begin{equation}
 \dot{e}_y=-K_P e_y+e(y,t) e_{z},~\dot{e}_{z} =- K_I e(y,t) {e}_y +\dot{z}_e, \label{eq:SyDyn3}
 \end{equation}
 which are easily stated in quasi-linear form. Lastly we only need to introduce the new disturbance input gain as defined in Prop.~\ref{prop1}, which takes the error-system states into account. \qed
 \end{pf}
Proposition \ref{prop1} will be used in Sec. \ref{sec:error_bounds} to define the stability and error bounds of the controlled system under usage of LMIs that will be derived in Sec.~\ref{sec:usage_of_bounds}.

\section{Stability and Error Bounds}
\label{sec:error_bounds}
In this section we consider the general strictly proper nonlinear system, representable in a quasi-linear form
\begin{equation}
\dot{\bm{x}}(t)=\bm{A}(\bm{x},t)\bm{x}(t)+\bm{b}u(t),~ y(t)=\bm{c}^T \bm{x}(t) \label{eq:SystemLMI}
\end{equation}
with state vector $\bm{x}(t)\in \mathbb{R}^{n_x}$, scalar input $u(t)$ and scalar output $y(t)$. The input- and output gain vectors $\bm{b}$ and $\bm{c}^T$ are assumed to be constant, while the state matrix  $\bm{A}(\bm{x}(t),t)$ may be state- and time-dependent but bounded by means of $L$ vertices $\bm{A}_l$ of a corresponding polytope. If a joint Lyapunov function holds simultaneously for all $L$ vertices $\bm{A}_l$ of a polytopic system, then it also holds for the convex combination of the vertices
\begin{equation}
\bm{A}(\bm{x},t)=\lambda_1(t) \bm{A}_1+ \dots + \lambda_L(t) \bm{A}_L\,, \lambda_l(t)\geq 0 \,, \sum_{l=1}^L{\lambda_l}=1 \,, \label{eq:Polytope}
\end{equation}
see e.g. \cite[pp.51-53]{Boyd:1994}.

\begin{prop}
\label{prop2}
If the following Linear Matrix Inequalities (LMIs) are satisfied simultaneously for all vertices $\bm{A}_l$ of a polytopic description of \eqref{eq:SystemLMI} for a given $\delta>0$
\begin{equation}
\left[\begin{matrix}
-1 & \bm{b}^T\bm{P} \\
\bm{P}\bm{b} & \bm{A}_l^T  \bm{P}+ \bm{P}\bm{A}_l+\delta  \bm{P} 
\end{matrix}\right]  \prec 0 \label{eq:LMI1}
\end{equation}
and
\begin{equation}
\bm{P}\succ 0,  \label{eq:LMI2}
\end{equation}
with $\bm{P}=\bm{P}^T$, the quadratic Lyapunov function
\begin{equation}
V(\bm{x}(t))=\bm{x}(t)^T \bm{P} \bm{x}(t) \label{eq:Ljap}
\end{equation}
for system \eqref{eq:SystemLMI} satisfies the following inequality at all times
\begin{equation}
V(x(t))< exp(-\delta t)  V(\bm{x}(0)) +\norm{u(t)}{\infty}^2/\delta \,. \label{eq:ie1}
\end{equation}
\end{prop}

 \begin{pf}
 The time derivative of the quadratic Lyapunov function \eqref{eq:Ljap} for system \eqref{eq:SystemLMI} is given by 
 \begin{equation}
 \dot{V}(t)=(\bm{x}^T\bm{A}(\bm{x},t)^T  +u\bm{b}^T ) \bm{P} \bm{x}+ \bm{x}^T\bm{P} (\bm{A}(\bm{x},t) \bm{x} +\bm{b} u),
 \end{equation}
 which can be reformulated as follows 
 \begin{equation}
 \dot{V}=\bm{x}^T(\bm{A}(\bm{x},t)^T  \bm{P}+  \bm{P}\bm{A}(\bm{x},t))\bm{x}+ 2 u\bm{b}^T\bm{P}\bm{x}\,. \label{eq:LjapTD}
 \end{equation}
 Due to Youngs inequality,
 \begin{equation}
 2 u\bm{b}^T\bm{P}\bm{x} \leq u^2 +(\bm{b}^T\bm{P}\bm{x})^2 = u^2 +\bm{x}^T\bm{P}\bm{b}\bm{b}^T\bm{P}\bm{x} \label{eq:Young}
 \end{equation} 
 holds. Consequently, we can conservatively approximate \eqref{eq:LjapTD} by means of
 \begin{equation}
 \dot{V}\leq\bm{x}^T(\bm{A}(\bm{x},t)^T  \bm{P}+  \bm{P}\bm{A}(\bm{x},t))\bm{x}+ u^2 +\bm{x}^T\bm{P}\bm{b}\bm{b}^T\bm{P}\bm{x},
 \end{equation}
 which can be rewritten as 
 \begin{equation}
 \dot{V}\leq\underbrace{\bm{x}^T(\bm{A}(\bm{x},t)^T  \bm{P}+  \bm{P}\bm{A}(\bm{x},t)+\bm{P}\bm{b}\bm{b}^T\bm{P})\bm{x}}_{\dot{{V}}_1}+ u^2  \,. \label{eq:iq2} 
 \end{equation}
Let us now consider the artificially destabilized system,
 \begin{equation}
 \dot{\bm{x}}(t)= \tilde{\bm{A}}(\bm{x},t) \bm{x}(t)+ \bm{b} u(t), \label{eq:SysDesta}
 \end{equation}
 with $\tilde{\bm{A}}(\bm{x},t)=\bm{A}(\bm{x},t) +\delta/2 \bm{I}^{n_x}$ and $\delta>0$. An equivalent quadratic Lyapunov function $\tilde{V}=\bm{x}^T \bm{P} \bm{x}$ for system \eqref{eq:SysDesta} under consideration of Youngs inequality \eqref{eq:Young} leads to the inequality
 \begin{equation}
 \dot{\tilde{V}}\leq\underbrace{\bm{x}^T(\bm{A}(\bm{x},t)^T  \bm{P}+  \bm{P}\bm{A}(\bm{x},t)+\bm{P}\bm{b}\bm{b}^T\bm{P}+\delta  \bm{P})\bm{x}}_{\dot{\tilde{V}}_1} +u^2 \,.
 \end{equation}
 Now, it becomes obvious that $\dot{{V}}_1=\dot{\tilde{V}}_1-\delta \bm{x}^T\bm{P}\bm{x}$ holds. By substituting this relation into \eqref{eq:iq2} and rearranging the terms, we can derive the inequality $\dot{V}+\delta \bm{x}^T\bm{P}\bm{x}-u^2\leq \dot{\tilde{V}}_1$. If we can prove that 
 \begin{equation}
 \dot{\tilde{V}}_1<0 ~\mathrm{for}~\bm{x}\neq\bm{0} \,, \label{eq:ie3}
 \end{equation}
 then  $\dot{V}+\delta \bm{x}^T\bm{P}\bm{x}-u^2<0$ for $\bm{x}\neq\bm{0}$ must hold as well. This can be rearranged in the form 
 \begin{equation}
 \dot{V}<-\delta \underbrace{\bm{x}^T\bm{P}\bm{x}}_{V}+u^2\,. \label{eq:ie4}
 \end{equation}
 Due to the linear first-order dynamics, the solution for this inequality \eqref{eq:ie4} can be derived analytically and provides inequality \eqref{eq:ie1}.

 Now, we still need to derive the LMIs to check \eqref{eq:ie3} for system \eqref{eq:SysDesta}. Inequality \eqref{eq:ie3} is guaranteed by finding a $\bm{P}=\bm{P}^T\succ 0$, which fulfills the matrix inequality 
 \begin{equation}
 \bm{A}(\bm{x},t)^T  \bm{P}+  \bm{P}\bm{A}(\bm{x},t)+\bm{P}\bm{b}\bm{b}^T\bm{P}+\delta  \bm{P} \prec 0 \label{eq:LMI4}
 \end{equation}
 for $\bm{A}(\bm{x},t)$. Using the polytopic system description from \eqref{eq:Polytope}, we consider the problem of finding $\bm{P}=\bm{P}^T\succ 0$ that simultaneously fulfills the  matrix inequalities for all $L$ vertices of the polytope
 \begin{equation}
 \bm{A}_l^T  \bm{P}+  \bm{P}\bm{A}_l+\bm{P}\bm{b} 1 \bm{b}^T\bm{P}+\delta  \bm{P} \prec 0, \label{eq:LMI5}
 \end{equation}
 which guarantees \eqref{eq:LMI4} under the given assumptions. These matrix inequalities can be reformulated to \eqref{eq:LMI1} using the Schur complement formula. \qed
 \end{pf}

\begin{prop}
The maximization of a linear function $y=\bm{c}^T\bm{x}$ with ellipsoidal boundaries $h(\bm{x})=\bm{x}^T\bm{P}\bm{x}-V$ is given by 
\begin{equation}
\norm{y}{\infty} =\underset{\bm{x}}{\mathrm{max}} \{\bm{c}^T \bm{x}|h(\bm{x})\leq 0\}=\sqrt{\bm{c}^T \bm{P}^{-1} \bm{c} \cdot{V}} \,. \label{eq:EllipsMax}
\end{equation}
\end{prop}
 \begin{pf}
 The proof is derived e.g. in \cite{Groetschel:1993}. \qed
 \end{pf}

\begin{prop} 
\label{prop:exstab}
For every system \eqref{eq:SystemLMI} that fulfills the LMIs given in Prop. \ref{prop2}, the following output bounds hold for all times $t>0$
\begin{equation}
y^2(t) \leq \bm{c}^T \bm{P}^{-1} \bm{c} \cdot \left(exp(-\delta t)  \bm{x}(0)^T\bm{P}\bm{x}(0) +\norm{u(t)}{\infty}^2/\delta \right) \,. \label{eq:p2p_full}
\end{equation}
\end{prop}
 \begin{pf}
 The proof follows directly from the fact that the Lyapunov function \eqref{eq:Ljap} is chosen in a standard quadratic form and, hence, the substitution of \eqref{eq:ie1} into the squared form of \eqref{eq:EllipsMax} is allowed.  \qed
 \end{pf}

\begin{rem}
\label{rem:remark1}
Since Prop. \ref{prop:exstab} holds for arbitrary output vectors $\bm{c}$, any system fulfilling the assumptions in Prop. \ref{prop2}  will also guarantee exponential stability of its states  for vanishing inputs $u=0$.
\end{rem}

\subsection{Peak-to-Peak Norm}
In accordance with \cite[p.79]{scherer:2015}, we define the so-called peak-to-peak (p2p) norm as follows
\begin{equation}
\norm{y(t)}{\infty}\leq\gamma \norm{u(t)}{\infty}  \,.\label{eq:p2p_norm}
\end{equation}
Here, $\norm{u(t)}{\infty}$ and $\norm{y(t)}{\infty}$  denote the upper bounds of the absolute values of the input and the output, respectively. This norm is different from the more prevalent $H_{\infty}$ norm that considers integral terms instead.

\begin{prop} 
\label{prop:p2p}
The following holds for system \eqref{eq:SystemLMI} with vanishing initial values $\bm{x}(0)=\bm{0}$:  If the  LMIs 
\begin{equation}
\left[\begin{matrix}
-1 & \bm{b}^T\bm{P} \\
\bm{P}\bm{b} & \bm{A}_l^T  \bm{P}+ \bm{P}\bm{A}_l+\delta  \bm{P} 
\end{matrix}\right]  \prec 0 \,,\label{eq:LMI6}
\end{equation}
\begin{equation}
\left[\begin{matrix}
\overline{\gamma} & \bm{c}^T \\
\bm{c} & \bm{P} 
\end{matrix}\right]  \succ 0 \label{eq:LMI7}
\end{equation}
and
\begin{equation}
\bm{P}\succ 0 \, \label{eq:LMI7_5}
\end{equation}
are simultaneously satisfied for a given $\delta>0$ for all $L$ vertices of the polytopic system description \eqref{eq:Polytope} under minimization of $\overline{\gamma}^*=\mathrm{min}\{\overline{\gamma}\}$ the system \eqref{eq:SystemLMI} satisfies the  p2p norm \eqref{eq:p2p_norm} with the minimal gain $\gamma=\sqrt{\overline{\gamma}^*/\delta}$, and the linearizing substitution $\overline{\gamma}=\gamma^2 \delta$.
\end{prop}

 \begin{pf}
 The LMIs \eqref{eq:LMI6} and \eqref{eq:LMI7_5}, which are identical to the LMIs  \eqref{eq:LMI1} and   \eqref{eq:LMI2}, guarantee that \eqref{eq:p2p_full} holds. The inequality \eqref{eq:p2p_full} simplifies to $y^2(t) \leq \bm{c}^T \bm{P}^{-1} \bm{c} \cdot \norm{u(t)}{\infty}^2/\delta$ for vanishing initial values and can be reformulated to 
  \begin{equation}
 |y(t)| \leq \underbrace{\sqrt{\bm{c}^T \bm{P}^{-1} \bm{c}/\delta}}_{\gamma}  \cdot \norm{u(t)}{\infty}, \label{eq:p2p2}
 \end{equation}
 which guarantees the fulfillment of the p2p norm \eqref{eq:p2p_norm} with
 \begin{equation}
 \gamma=\sqrt{\bm{c}^T \bm{P}^{-1} \bm{c}/\delta} \,. \label{eq:gamma}
 \end{equation}
 One wants, however, to find the minimal value $\gamma>0$ that fulfills the LMIs. For this purpose, $\overline{\gamma}=\gamma^2 \delta$ is introduced, which allows us to reformulate \eqref{eq:gamma} to the equation $\overline{\gamma}=\bm{c}^T \bm{P}^{-1} \bm{c}$. For a given $\delta>0$, any minimization of $\overline{\gamma}>0$ will also minimize $\gamma>0$. Since $\overline{\gamma}>0$ holds, we can state $\bm{c}^T \bm{P}^{-1} \bm{c}-\overline{\gamma}>0$ which can be reformulated in the form of the respective LMI \eqref{eq:LMI7} using the Schur complement formula. \qed
 \end{pf}

\begin{rem}
Since the LMIs are not linear in $\delta$, it cannot be optimized directly.  In our case, finding the value $\delta$ that minimizes $\gamma$ is done via bisection. For the solution of the LMIs, we use YALMIP, see \cite{Lofberg:2004}, together with SeDuMi, see \cite{sturm:1999}.
\end{rem}

\begin{rem}
Due to the use of Youngs inequality and the nonlinear parameter $\delta$, the p2p norm may be quite conservative, which is also mentioned by \cite[p.99]{scherer:2015}. The $H_\infty$ serves as lower bound on the p2p norm and is generally easier to evaluate. 
\end{rem}

\begin{rem}
For the linear case, for strictly proper SISO systems, the LMIs \eqref{eq:LMI6} to \eqref{eq:LMI7_5} are equivalent to the LMI presented in \cite[p.80]{scherer:2015}, which can be shown analytically. A further equivalence is found by comparison with \cite[pp.83-84]{Boyd:1994} and \cite[p.88]{Boyd:1994}.  Our LMI and the LMI variant presented in \cite[p.80]{scherer:2015} were compared w.r.t.~computational speed and final results -- there was no noticeable difference.
\end{rem}

\section{Application of the Error Bounds}
\label{sec:usage_of_bounds}
Now, we want to utilize the previously derived output and state bounds for our particular controlled system \eqref{eq:QlinSysDyn}.

\begin{thm} 
\label{thm:appl1}
If the disturbance input gain $\overline{e}(e_y(t),y_r,t)$ can be bounded by $e^-\leq \overline{e}(e_y(t),y_r,t)\leq e^+$  for a given $\delta$, the simultaneous feasibility of the LMIs

\begin{equation}
\left[\begin{matrix}
-1 & \bm{b}^T\bm{P} \\
\bm{P}\bm{b} & \bm{A}_1^T  \bm{P}+ \bm{P}\bm{A}_1+\delta  \bm{P} 
\end{matrix}\right]  \prec 0  \,,\label{eq:LMI8}
\end{equation}
and
\begin{equation}
\left[\begin{matrix}
-1 & \bm{b}^T\bm{P} \\
\bm{P}\bm{b} & \bm{A}_2^T  \bm{P}+ \bm{P}\bm{A}_2+\delta  \bm{P} 
\end{matrix}\right]  \prec 0 \,, \label{eq:LMI9}
\end{equation}
as well as the LMIs
\begin{equation}
\left[\begin{matrix}
\overline{\gamma} & \bm{c}^T \\
\bm{c} & \bm{P} 
\end{matrix}\right]  \succ 0 \,, \label{eq:LMI10}
\end{equation}
and
\begin{equation}
\bm{P}\succ 0  \label{eq:LMI11}
\end{equation}
under minimization of $\overline{\gamma}^* =\mathrm{min}\{\overline{\gamma}\}  \,,$ with the polytopic vertices
\begin{equation}
\bm{A}_1=\left[ \begin{matrix}
-K_P& e^+ \\
-K_I e^+&0 	
\end{matrix}\right] \,, \bm{A}_2=\left[ \begin{matrix} \label{eq:poly2}
-K_P& e^- \\
-K_I e^-&0 
\end{matrix}\right] \,,
\end{equation}
and $\bm{b}=[0,1]^T$, $\bm{c}^T=[1,0]$ guarantees that the system \eqref{eq:SystemBase} controlled with \eqref{eq:ControlLaw} and \eqref{eq:ParamUpdateLaw} fulfills the following p2p norm for vanishing initial conditions
\begin{equation}
\norm{e_y(t)}{\infty}\leq\gamma \norm{\dot{z}_e(t)}{\infty} \,, \label{eq:p2p_norm2}
\end{equation}
with $\gamma=\sqrt{\overline{\gamma}^*/\delta}$. 
\end{thm}

 \begin{pf}
 The reformulated error-system \eqref{eq:QlinSysDyn} fulfills the conditions for system \eqref{eq:SystemLMI}. The equations  \eqref{eq:LMI8} to \eqref{eq:p2p_norm2} are just a straightforward implementation of Prop. \ref{prop:p2p}. \qed
 \end{pf}

\begin{cor}
\label{cor:appl2}
If the LMIs of Thm. \ref{thm:appl1} are fulfilled, the system \eqref{eq:SystemBase} controlled with \eqref{eq:ControlLaw} has a bounded output error for non-vanishing initial conditions
\begin{equation}
e_y(t) \leq \sqrt{ \bm{c}^T \bm{P}^{-1} \bm{c} \cdot \left(exp(-\delta t)  \bm{x}(0)^T\bm{P}\bm{x}(0) +\norm{\dot{z}_e(t)}{\infty}^2/\delta \right)}\,. \label{eq:p2p_full_apl} 
\end{equation}
\end{cor}

 \begin{pf}
 The bound  \eqref{eq:p2p_full_apl} is an application of the bound \eqref{eq:p2p_full}. The related LMIs \eqref{eq:LMI1} and \eqref{eq:LMI2} are implemented as the LMIs \eqref{eq:LMI8}, \eqref{eq:LMI9}  and \eqref{eq:LMI11}. They are thus automatically fulfilled with the LMIs of Thm. \ref{thm:appl1}. \qed
 \end{pf}

\begin{cor}
\label{cor:appl3}
If the LMIs of Thm. \ref{thm:appl1} are fulfilled, the system \eqref{eq:SystemBase} controlled with \eqref{eq:ControlLaw} and \eqref{eq:ParamUpdateLaw} has a bounded parameter error for non-vanishing initial conditions
\begin{equation}
e_z(t) \leq \sqrt{ \bm{c}_z^T \bm{P}^{-1} \bm{c}_z \cdot \left(exp(-\delta t)  \bm{x}(0)^T\bm{P}\bm{x}(0) +\norm{\dot{z}_e(t)}{\infty}^2/\delta \right)}\,, \label{eq:p2p_full_aplz} 
\end{equation}
with $\bm{c}_z^T=[0,1]$. 
\end{cor}
 \begin{pf}
 See Cor.  \ref{cor:appl2} and Remark \ref{rem:remark1}. \qed
 \end{pf}

\begin{cor}
\label{cor:appl4}
If  the LMIs of Thm. \ref{thm:appl1} are fulfilled and  GP prediction errors are constant, e.g. $v=\dot{z}_e=0$, system \eqref{eq:SystemBase} controlled with \eqref{eq:ControlLaw} and \eqref{eq:ParamUpdateLaw}  is exponentially stable. Therefore, $e_y(t)\rightarrow0$ and  $e_z(t)\rightarrow0$  hold for $t\rightarrow\infty$.
\end{cor}
 \begin{pf}
 The vanishing errors follow as a special case of \eqref{eq:p2p_full_apl} and \eqref{eq:p2p_full_aplz} for $\norm{\dot{z}_e}{\infty}=0$. \qed
 \end{pf}

\begin{rem}
Cor. \ref{cor:appl4} guarantees perfect tracking for $t\rightarrow\infty$ if $\dot{z}_e=0$ holds. The vanishing prediction error rate $\dot{z}_e=0$ with $z_e=\hat{z}-z$ is achieved, e.g., if the GP prediction is perfect or possesses a constant offset or  if the learning is converged and  $\dot{\bm{\zeta}}=\bm{0}$.
\end{rem}

\subsection{Proposed Utilization of the Bounds}
\label{sec:ProposedUtilization}
A possible workflow can be as follows: One could define a bound on $y_r \in Y_r =[y_r^-,y_r^+]$   for the system. This bound can be used to conservatively approximate the bounds of the disturbance input gain by means of some desired maximum output error $e_{y,lim}$, e.g.  $e^-=\mathrm{min}\{\overline{e}|y_r \in Y_r,|e_y(t)|\leq e_{y,lim}\}$ and $e^+=\mathrm{max}\{\overline{e}|y_r \in Y_r,|e_y(t)|\leq e_{y,lim}\}$. Then, the LMIs from Thm. \ref{thm:appl1} are evaluated to derive the p2p gain  $\gamma$ of the error system. These can then be used to define a bound 
\begin{equation}
z_{lim}=\mathrm{max}(e_{y,lim}/\gamma -\norm{\dot{z}}{\infty},0) \label{eq:DERLbound}
\end{equation}
for the DERL according to Cor. \ref{cor:DERL} if $\norm{\dot{z}}{\infty}$ can be approximated.

\begin{rem}
Besides the necessary mixture of continuous-time and discrete-time formulations as indicated in Def.~\ref{def:DERL}, the control algorithm contains several assumptions. While the GPSOL algorithm works very well with noisy measurements -- for the proposed inversion-based control algorithm and the derived error system -- we assume a perfect knowledge of the system output $y(t)$. For the control itself, this is no large hindrance as this combination of stochastic filters and deterministic control is a very popular choice and does -- in particular cases -- even lead to the optimal solution. While so far there have not been any problems in practice, one should, however, be careful with the derived output bounds in the case of very noisy measurements, because they do not take into account the uncertainty w.r.t.~the inversion-based control.
\end{rem}

\section{Simulation Results for the Error Bounds}
\label{sec:numval}
Firstly, we want to demonstrate the calculation of the output bound for the following nonlinear system
\begin{equation}
\dot{y}(t)=\underbrace{-1+y^3(t)}_{a(y(t))}+\underbrace{y^3(t)}_{b(y(t))}  u(t)+\underbrace{y(t)}_{e(y(t))}  z(t)   \,.
\end{equation}
The system is controlled by means of the control law \eqref{eq:ExControlLaw} and the parameter update law \eqref{eq:ParamUpdateLaw}.  We chose the control gains as follows: $K_I=20$ and $K_P=10$. The bounds of the  disturbance input gain are defined as  $e \in [1,10]$. This provides us with the following vertices of the polytopic description of the controlled system 
\begin{equation}
A_1=\left[ \begin{matrix}
-10& 10 \\
-200&0 
\end{matrix}\right],~ A_2=\left[ \begin{matrix}
-10& 1 \\
-20 &0 
\end{matrix}\right] \,. \label{eq:simvalpoly}
\end{equation}
For the simulation, we define the reference value as a uniform distributed signal: $y_r\sim \mathcal{U}(1.11,9.89)$, with a sampling time $T_y=0.9~\mathrm{s}$, which is low-pass filtered with a time constant $T_1=1e-3~\mathrm{s}$ to allow for an exact differentiation.

As at this point only the bounds shall be tested, the output of the hidden function $z(t)$ and it's estimate $\hat{z}(t)$ are defined by means of two independent Gaussian variables $\mathcal{N}(0,100^2)$, with a sampling time $T_z=3~\mathrm{s}$. The signals are rate limited: $\norm{\dot{z}}{\infty}< 0.2$ and $\norm{\dot{\hat{z}}}{\infty}< 0.2$.

The computation of the gain $\gamma$ by using \eqref{eq:p2p_norm2} for \eqref{eq:simvalpoly} results in  $\gamma=0.2653$. According to the definition of this gain in \eqref{eq:p2p_norm2}, the resulting output error $e_y=y_r-y$ should be bounded by 
\begin{equation}
\norm{e_{y}}{\infty}=(\norm{\dot{z}}{\infty}+\norm{\dot{\hat{z}}}{\infty}) \gamma =0.4 \cdot0.2653= 0.1061 \,.
\end{equation}

The simulation of $500s$ results   in $\norm{e_{y}}{\infty}=0.03328$, which is within the computed bound. For a comparison, the infinity gain results in $\gamma_{\infty}=0.0756$. This would  lead to  the output error bound 
\begin{equation}
\norm{e_{y}}{\infty}=(\norm{\dot{z}}{\infty}+\norm{\dot{\hat{z}}}{\infty}) \gamma_\infty =0.4 \cdot0.0756= 0.03024, 
\end{equation}
which is clearly violated in this particular case.

\section{Experimental Results}
\label{sec:exval}
To validate the complete control algorithm, we utilize the pneumatic test rig depicted in Fig.~\ref{pic:PneumPruefstand}. The system consists of two pneumatic valves, a pneumatic tank of volume $V=0.4~\mathrm{l} $ and a pressure sensor. The first pneumatic valve represents a 5/3 way valve, which is able to fill or deplete the tank based on the supply pressure $p_{in}=4~\mathrm{bar}$  or the ambient pressure $p_U$, respectively. The second valve is only employed as a throttle with variable diameter. The first valve possesses a known characteristic, and its voltage is the control input $u_C$, whereas the second valve characteristic is unknown. The corresponding control voltage $\zeta_1=u_z$ represents a measurable disturbance. The control system is implemented on a Bachmann real time hardware (CPU:MH230) with a sampling time of $T_s=1~\mathrm{ms}$.

\begin{figure}
	 \begin{center}
		 \includegraphics[width=0.8\linewidth]{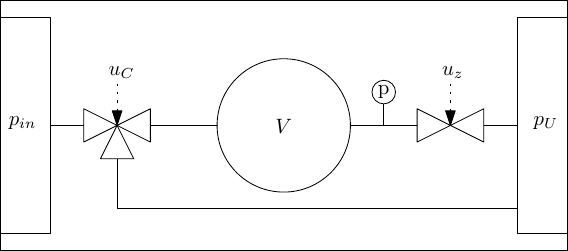}
		  \caption{Scheme of the considered pneumatic test rig at the Chair of Mechatronics, University of Rostock.} 
		  \label{pic:PneumPruefstand}
	 \end{center}
\end{figure}

\subsection{Model-Based Controller}
The pneumatic system can be described by a mass balance
\begin{equation}
\dot{m}(t)=\dot{m}_{in}(t)-\dot{m}_{out}(t) \label{eq:Mbalance} \,.
\end{equation}
Here, the air mass inside the tank is related to the tank pressure by means of the ideal gas law $p(t) V/(R T) =m(t) $. Assuming an isothermal thermodynamic process, the corresponding time derivative becomes $\dot{p}(t) V/(R T) =\dot{m}(t)$, and can be substituted in \eqref{eq:Mbalance} to derive $\dot{p}=R T/V \left(\dot{m}_{in}-\dot{m}_{out} \right) $.  The input mass flow depends in a known, nonlinear manner on the pressure drop over the first valve and its actuation voltage $\dot{m}_{in}(t)=g(p_{in},p_{U},p(t),u_C(t))$. The output mass flow $\dot{m}_{out}(t)=\Psi(p(t)-p_U)z_f(u_z(t))$ through the second valve is modeled by means of a pressure dependent term $\Psi(p-p_U)$ and another term depending on the actuation voltage $z_f(u_z)$, which represents the hidden function.
The overall nonlinear system dynamics can be represented in the following ODE
\begin{equation}
\dot{p}=\frac{R T}{V} g(p_{in},p_{U},p,u_C)-\left(\frac{R T}{V}\Psi(p-p_U)\right) z_f(u_z) \,.
\end{equation}
The relation of the outlet mass flow regarding the pressure drop is modeled by means of the following regularized turbulent flow equation
\begin{equation}
\Psi(p-p_U)=\mathrm{sgn}(p-p_U) \left(-\frac{\epsilon }{2}+\sqrt{\frac{\epsilon^2}{4}+|p-p_U|}\right) \,.
\end{equation}

An well-established method for the control of pneumatic systems, as also proposed in \cite{Wache:2021}, is an inversion of the input nonlinearities by evaluating $u_C=g_{I}(p_{in},p_{U},p,\dot{m}_{in,d})$ and defining $v=\dot{m}_{in,d}$ as a new input. This is applied here as well, which results in the following system 
\begin{equation}
\frac{dy}{dt}=\underbrace{\frac{R T}{V}}_b v+\underbrace{-\frac{R T}{V} \Psi(y-p_U)}_{e(y)} z_f(u_z)  \,.
\end{equation}
with state $y=p$. The controller is consists of the control law \eqref{eq:ExControlLaw} as well as the adaptation law \eqref{eq:ParamUpdateLaw}. The GPSOL algorithm, see Sec. \ref{sec:GPSOL}, is used to learn a model $\tilde{z}_f(u_z(t),(t))$ for the hidden function $z_f(u_z(t))$. The output of the GP prediction is filtered using the DERL filter, see Def.~\ref{def:DERL}. This results in a filtered prediction $\hat{z}(t)$ according to Def.~\ref{def:DERL}. The rate limit $z_{lim}$ is computed with  \eqref{eq:DERLbound} with  $\dot{z}=0$ and different $e_{y,lim}$ as specified later.  The range of the pressure drop over the outlet valve is approximated by $p-p_U\in[1,3]~\mathrm{bar}$. Together with the control gains, this leads to the polytopic system 
\begin{equation}
A_1=\left[ \begin{matrix}
-4& -0.0665 \\
66.5359&0 
\end{matrix}\right] \,,A_2=\left[ \begin{matrix}
-10& -0.1152 \\
115.2444 &0 
\end{matrix}\right] \,. \label{eq:exvalpoly}
\end{equation}
After an evaluation of the LMIs according to Theorem \ref{thm:appl1}, we obtain a p2p gain $\gamma=0.024$. The remaining parameters and hyperparameters are summarized in Table. \ref{table:Hyperparamters_ex}.

  \begin{table}
\centering
\captionsetup{width=0.9\hsize}
  \caption{Experiment Parameters and Hyperparameters}
\resizebox{\columnwidth}{!}{%
\begin{tabular}{ c c c c | c c c c } 
\hline
\multicolumn{4}{c|}{System and Control Parameters} & \multicolumn{4}{c}{GPSOL Hyperparameters}\\
\hline
 $T_s$  & $1e-3~\mathrm{s}$ & $R$  & $0.2871~\mathrm{\frac{J}{g K}}$ &$n_z$  & $1$ & $N_1$  & $6$ \\
 $T_0$  & $293.15~\mathrm{K}$ & $V$  & $4e-4~\mathrm{m^3}$  & $\underline{\zeta}_1$  & $0~\mathrm{V}$& $\overline{\zeta}_1$ & $5~\mathrm{V}$\\
 $K_P$  & $4$  &  $K_I$  & $1000$ &$L$ & $2.5$ & $\sigma_K$  & $1$ \\
 $K_\sigma$  & $5$  &  $ $ &  &  $Q_x$ & $(1e-5)^2$ & $R$ & $(1e-5)^2$\\
 \hline
  \end{tabular}%
 }
\label{table:Hyperparamters_ex}
\end{table}

\subsection{Results for the Error Bounds}
Firstly, we want to present the functionality of the DERL filter and the validity of the presented bounds. Therefore, we set the outlet valve to different stationary outlet valve openings with $\dot{z}=0$. Thereby, inequality \eqref{eq:p2p_norm2} is reduced to $\norm{e_y}{\infty}\leq \gamma z_{lim}$, which allows for a direct validation of the bounds (see Cor. \ref{cor:DERL}). We conduct the test for different pressure reference values of  $y_r=2.7~\mathrm{bar}$ and $y_r=1.3~\mathrm{bar}$. Moreover, the valve openings correspond to $u_z=1~\mathrm{V}$ and $u_z=4~\mathrm{V}$. The desired bound is defined as $e_{y,lim}=0.05~\mathrm{bar}$. 

The maximum error for all tests with the enabled DERL is $\norm{e_{y}}{\infty}=0.029~\mathrm{bar}$, which does not breach the bound.  For a comparison, without the DERL a maximum error of $\norm{e_{y}}{\infty}=0.45~\mathrm{bar}$ is measured. Therefore, also the experimental validation of the bounds can be considered as successful.

\subsection{Results for the Complete Control Structure}
To reduce the impact of stochastic effects, the validation of the complete control structure is conducted as follows: A filtered trajectory based on random steps for the desired pressure is defined for 5 runs with 100s each. Within this time span, equally filtered trajectories stemming from random steps with a higher sample rate are defined for $u_z$. All 5 runs are evaluated for all variants and algorithms with a settling time of 20s between the runs. Moreover, the GPs are re-initialized after all single runs. 

In Fig.~\ref{pic:CumError_stat2}, we compare the cumulated absolute output error (CAE) for all 5 runs for each variant, normalized to the results without GPSOL. It becomes obvious that the result with GPSOL massively outperforms the baseline -- the CAE is reduced by more than $74\%$. Furthermore, it turns out that the CAE for the control variants with an application of the DERL filter  shows lower errors at the beginning, but in general larger overall errors, which leads to a lower rise of the CAE in the beginning, but to higher end-values. This behavior is a direct result of the DERL filter. In order to achieve the output error bounds, the application of the GPSOL model prediction has to be filtered, which naturally reduces the performance. This trade-off between performance and robustness is actually seen as a function of the defined bounds $e_{y,lim}$. The closer the bounds are defined, the slower the learned model gets applied, which generally reduces the overall performance.

\begin{figure}
	 \begin{center}
		 \includegraphics[width=1\linewidth]{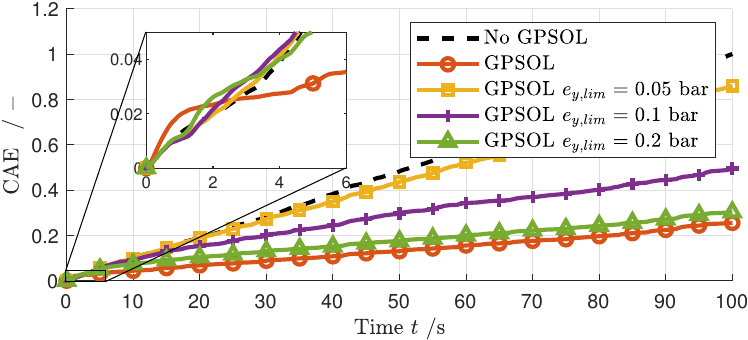}
		  \caption{Cumulated absolute error (CAE) for different bounds $e_{y,lim}$, averaged over 5 test runs.} 
		  \label{pic:CumError_stat2}
	 \end{center}
\end{figure}

\section{Conclusion and Outlook}
This paper presents a learning-based control for a particular class of systems with a set of stability conditions and output error bounds in dependence on the GP model prediction error rate. The stability conditions and the bounds can be efficiently calculated by using linear matrix inequalities (LMIs). The algorithm builds upon a previously presented GPSOL algorithm as well as the DERL filter. The GPSOL algorithm enables the online learning of an approximate model for an unknown subsystem, whereas the DERL filter provides a chance constraint on the prediction error rate of the learned GP. If the assumptions are satisfied, the complete control structure leads to an exponentially stable learning-based controller with guaranteed output error bounds. The derived bounds are validated in simulations of a nonlinear system. Furthermore, we present an experimental validation of the bounds as well as the complete control structure for a pneumatic test rig.

A natural extension of the controller and the derived bounds would involve the consideration of higher-order and multiple-input multiple-output (MIMO) systems. Under certain assumptions, these extensions should be possible. There could, however, emerge further problems regarding the inherent conservativeness of the bounds, which might increase. In \cite{GPSOLx:2026}, also submitted to the IFAC World Congress 2026, we present an extension of the GPSOL algorithm to hidden functions with uncertain inputs, which enables the consideration of system states as GP inputs. The combination of these two methods would be straightforward and would greatly increase the applicability of the proposed control.

\bibliography{ifacconf,myrefs}

\end{document}